%% file: main.tex
\documentclass[11pt,a4paper,twoside]{article}
\usepackage[T1]{fontenc}
\usepackage[margin=2.8cm,headheight=14pt,headsep=18pt]{geometry}
\usepackage{amsmath,amssymb,amsthm,mathtools}
\usepackage{graphicx,booktabs,array,microtype,fancyhdr}
\usepackage[font=small,labelfont=bf,hypcap=false]{caption}
\usepackage{enumitem}
\usepackage{xurl}
\usepackage[hidelinks]{hyperref}
\hypersetup{pdftitle={Market Completeness and Optional Projections under Restricted Information},pdfauthor={Levin David Schwab},pdfcreator={},pdfproducer={},pdfsubject={Finite discrete-time markets, optional projections, and completeness},pdfkeywords={}}
\setlist{nosep,leftmargin=1.5em}
\theoremstyle{plain}
\newtheorem{proposition}{Proposition}[section]
\newtheorem{theorem}[proposition]{Theorem}
\newtheorem{lemma}[proposition]{Lemma}
\newtheorem{corollary}[proposition]{Corollary}
\theoremstyle{definition}
\newtheorem{definition}[proposition]{Definition}

\newtheorem{example}[proposition]{Example}
\numberwithin{equation}{section}
\newcommand{\E}{\mathbb E}
\newcommand{\R}{\mathbb R}
\newcommand{\F}{\mathcal F}
\newcommand{\G}{\mathcal G}
\newcommand{\M}{\mathcal M}
\newcommand{\N}{\mathcal N}
\newcommand{\ind}{\mathbf 1}
\newcommand{\sbar}{\overline S}
\newcommand{\A}{\mathcal A}
\begin{document}
\raggedbottom
\begin{center}
{\Large\bfseries Market Completeness and Optional Projections\\[2pt] under Restricted Information\par}
\vspace{12pt}
{\large Levin David Schwab\par}
\vspace{4pt}
22 September 2026
\end{center}
\thispagestyle{plain}
\vspace{3pt}
\begin{abstract}
\noindent In a finite discrete-time market, trading decisions may be predictable with respect to a filtration that does not adapt asset prices. The first fundamental theorem then characterizes absence of arbitrage by measures under which the optional projection of discounted prices is a martingale. We examine the corresponding completeness question. For a fixed projection, every claim measurable with respect to its terminal price history is attainable precisely when all equivalent martingale measures of that fixed process agree on the claim sigma-field. We give the finite-dimensional proof, retaining the distinction between the trading filtration and the claim sigma-field. If the projection is common to all optional martingale measures of the original prices, projected completeness implies uniqueness of their restrictions, but the converse fails even in a three-state model with a unique optional martingale measure. For the binomial model under its product martingale measure, a delay of $k$ periods yields a complete projected market with effective horizon $(T-k)^+$. An explicit replication construction and finite examples distinguish this completeness from replication at the original prices.
\end{abstract}

\section{Introduction}
When an investor chooses positions using delayed or restricted information, the discounted price process need not be adapted to the trading filtration. Kabanov and Stricker~\cite{KS2006} showed that absence of arbitrage in this setting is equivalent to the existence of an equivalent probability measure under which the optional projection of prices is a martingale. The process being projected and the measure used to project it must both be specified.

The analogous completeness question has an additional difficulty. Fixing an optional martingale measure produces an adapted projected price process, to which the classical second fundamental theorem applies. Changing the measure, however, may change that process. Even if the projection is the same under every optional martingale measure of the original market, these measures need not exhaust the martingale measures of the fixed projected market. Consequently, uniqueness within the original class can be too weak to imply projected completeness.

We consider completeness for claims measurable with respect to the projected price history while allowing strategies to use the full trading filtration. For a fixed projection, completeness is equivalent to agreement of all its martingale measures on the claim sigma-field. Under a common-projection assumption, completeness implies agreement of the original optional martingale measures on that sigma-field. A three-state counterexample disproves the converse. In the delayed binomial model, the product martingale measure gives a complete projected market, with an explicit replication formula for both risky and cash holdings.

The contribution is a finite-state clarification of the relevant measure classes, illustrated by explicit counterexamples and a constructive delayed binomial model. No general priority claim is made for the fixed-process duality criterion or the replication method. The underlying tools are established. The finite-dimensional separation argument is the standard proof of the second fundamental theorem, as presented by B\"auerle and Rieder~\cite[Theorem~5.1]{BR2017}. Grigorian and Jarrow~\cite{GJ2024,GJ2025} study complete fictitious markets obtained by filtration reduction and the extension of their pricing measures to an original market. Their discrete models project under the statistical measure and impose conditions on extensions. Here the starting point is a measure supplied by the restricted-information first fundamental theorem, and the question concerns the distinction between two classes of measures for a fixed projection. Biagini, Mazzon and Perkki\"o~\cite{BMP2023} study optional projections under changes of equivalent local martingale measures in continuous time. Kardaras and Ruf~\cite{KR2020} show in continuous time that completeness can fail under filtration shrinkage and that projections of deflators need not exhaust the deflators of the smaller market. Their stock remains adapted to the smaller filtration; here the original prices may be nonadapted, and the comparison concerns probability measures for a fixed optional projection. Projected markets, measure-dependent projections, and failures of completeness under information reduction are therefore established topics.

The delayed binomial construction concerns exact replication in the projected market. It has a different target from super-replication at the original prices with delayed trading, studied by Ichiba and Mousavi~\cite{IM2017}. The completeness results have finite state space and finite time horizon. The appendices provide the supporting derivations, a complete proof of the known first fundamental theorem on arbitrary probability spaces, additional examples, and a complete mathematical hedging procedure. No continuous-time or infinite-state completeness assertion is made.

\section{Prices, information, and projected trading}\label{sec:setup}
Let $\Omega$ be finite, let $\F=2^\Omega$, and let $P(\{\omega\})>0$ for every $\omega\in\Omega$. Fix $T\in\mathbb N$ and a filtration $\mathbb F=(\F_t)_{t=0}^T$ with
\begin{equation}\label{eq:filtration}
 \F_0=\{\varnothing,\Omega\},\qquad \F_T\subseteq\F.
\end{equation}
Let $B_t>0$ be deterministic with $B_0=1$, and let $S_t\colon\Omega\to\R_+^d$ be finite-valued, with deterministic $S_0$. Write $\sbar_t=S_t/B_t$ for discounted prices. The variables $S_t$ are $\F$-measurable, but need not be $\F_t$-measurable. In this finite setting, all random variables and finite-valued strategies are bounded and integrable under every probability measure on $\Omega$.

A risky-asset strategy $H=(H_t)_{t=1}^T$ is $\mathbb F$-predictable if $H_t$ is $\F_{t-1}$-measurable. Denote this linear space of strategies by $\mathcal H$. The restricted-information gain space is
\begin{equation}\label{eq:originalgains}
 K_S=\left\{(H\cdot\sbar)_T:H\in\mathcal H\right\},\qquad
 (H\cdot\sbar)_t=\sum_{s=1}^t H_s\cdot(\sbar_s-\sbar_{s-1}).
\end{equation}
We use the no-arbitrage condition $K_S\cap\R_+^\Omega=\{0\}$, which is the gain-based convention of Kabanov and Stricker~\cite{KS2006}. Short positions and borrowing are permitted; there are no further portfolio constraints.

\subsection{Optional martingale measures}
For $Q\sim P$, define the discrete optional projection by
\begin{equation}\label{eq:projection}
 X_t^Q=\E_Q[\sbar_t\mid\F_t],\qquad 0\leq t\leq T,
\end{equation}
and set
\begin{equation}\label{eq:M}
 \M_S=\{Q\sim P:X^Q\text{ is an }(\mathbb F,Q)\text{-martingale}\}.
\end{equation}
Since $Q$ is strictly positive at every state, these conditional expectations are defined on every atom without a choice of null-set versions. The tower property gives
\begin{equation}\label{eq:momentcondition}
 Q\in\M_S
 \quad\Longleftrightarrow\quad
 \E_Q[\Delta\sbar_t\mid\F_{t-1}]=0\quad(1\leq t\leq T).
\end{equation}
Indeed, $\E_Q[X_t^Q\mid\F_{t-1}]-X_{t-1}^Q
=\E_Q[\sbar_t-\sbar_{t-1}\mid\F_{t-1}]$.
Testing~\eqref{eq:momentcondition} against indicators in $\F_{t-1}$ also shows that
\begin{equation}\label{eq:annihilatororiginal}
 \M_S=\{Q\sim P:\E_Q[g]=0\text{ for every }g\in K_S\}.
\end{equation}
The first fundamental theorem under restricted information states that
\begin{equation}\label{eq:firstftap}
 K_S\cap\R_+^\Omega=\{0\}\quad\Longleftrightarrow\quad\M_S\ne\varnothing.
\end{equation}
This is the finite-state instance of~\cite[Theorem~1]{KS2006}. Its general form also addresses closure and integrability on arbitrary probability spaces. Those additional issues do not arise here. In particular, every density is automatically bounded: $dQ/dP$ is a finite vector.

\subsection{The fixed projected market}
Assume $\M_S\ne\varnothing$, choose $Q_0\in\M_S$, and keep
\begin{equation}\label{eq:fixedX}
 X=X^{Q_0}
\end{equation}
fixed. Thus $X$ is an adapted discounted price process and a $Q_0$-martingale. Let
\begin{equation}\label{eq:G}
 \G_t=\sigma(X_0,\ldots,X_t)\subseteq\F_t
\end{equation}
be its natural filtration. Terminal claims will be nonnegative $\G_T$-measurable variables. Trading strategies remain $\mathbb F$-predictable. Thus strategies may use information in $\mathbb F$ beyond the projected price history.

Given deterministic initial capital $c$ and $H\in\mathcal H$, define discounted projected wealth and the holding in the money-market account by
\begin{align}
 W_t&=c+\sum_{s=1}^tH_s\cdot\Delta X_s,\label{eq:wealth}\\
 \beta_t&=W_{t-1}-H_t\cdot X_{t-1},\qquad 1\leq t\leq T.\label{eq:cash}
\end{align}
Then $\beta_t$ is $\F_{t-1}$-measurable and
\begin{equation}\label{eq:selffinance}
 \beta_t+H_t\cdot X_t=W_t
 =\beta_{t+1}+H_{t+1}\cdot X_t\qquad(1\leq t<T).
\end{equation}
Thus $(\beta,H)$ is a predictable self-financing strategy in the auxiliary market with money-market account $B$ and risky prices $B_tX_t$. Its undiscounted wealth is $B_tW_t$.

\begin{definition}\label{def:complete}
A nonnegative $\G_T$-measurable claim $\xi$ is \emph{attainable in the projected market} if $\xi/B_T=c+(H\cdot X)_T$ for some $c\in\R$ and $H\in\mathcal H$. The projected market is \emph{complete for $\G_T$-claims} if every such claim is attainable.
\end{definition}

Put $K_X=\{(H\cdot X)_T:H\in\mathcal H\}$ and $\A_X=\R\ind+K_X$. Completeness is equivalent to $L^0(\G_T)\subseteq\A_X$, since any real claim on a finite space becomes nonnegative after addition of a constant. The space $\A_X$ lies in $L^0(\F_T)$, but need not lie in $L^0(\G_T)$: a trading decision can use information in $\F_{t-1}$ absent from $\G_{t-1}$.

\subsection{The trading convention matters}
The cash construction~\eqref{eq:cash} is justified by adaptedness of $X$. For the original nonadapted prices, the formal self-financing cash balance
\begin{equation}\label{eq:originalcash}
 \beta_t^S=c+(H\cdot\sbar)_{t-1}-H_t\cdot\sbar_{t-1}
\end{equation}
need not be $\F_{t-1}$-measurable. Requiring both cash and risky holdings to be predictable therefore gives a stricter model than~\eqref{eq:originalgains}. The unrestricted enhancement of every predictable risky strategy to a predictable full portfolio is valid in the projected market, but cannot be transferred to the original prices without checking this measurability.

Likewise, projected wealth is not generally the conditional projection of original gain-based wealth. Earlier prices can become better known between their trading dates and maturity, so there is no general identity
\begin{equation}\label{eq:notcommute}
 \E_{Q_0}[(H\cdot\sbar)_T\mid\F_T]=(H\cdot X)_T.
\end{equation}
An explicit failure appears in Section~\ref{sec:binomialexample}. Definition~\ref{def:complete} is an exact replication statement about the auxiliary prices $B_tX_t$; it is not a claim that the same portfolio replicates at prices $S_t$.

\section{Completeness and the class of measures}\label{sec:criterion}
For the fixed process $X$ in~\eqref{eq:fixedX}, define
\begin{equation}\label{eq:N}
 \N_X=\{R\sim P:X\text{ is an }(\mathbb F,R)\text{-martingale}\}.
\end{equation}
Here $X$ stays fixed when $R$ changes. In particular, $R\in\N_X$ does not assert that $X=\E_R[\sbar\mid\mathbb F]$, nor that $R\in\M_S$. We have $Q_0\in\N_X$.

\begin{lemma}\label{lem:annihilator}
For $R\sim P$, membership in $\N_X$ is equivalent to $\E_R[g]=0$ for every $g\in K_X$. If $R\in\N_X$, every wealth process~\eqref{eq:wealth} is an $(\mathbb F,R)$-martingale.
\end{lemma}
\begin{proof}
If $X$ is an $R$-martingale, predictability gives
$\E_R[H_t\cdot\Delta X_t\mid\F_{t-1}]=0$ for each $t$. Summation proves both assertions. Conversely, test zero expected gains with the strategy that holds the $i$th asset only at time $t$, in amount $\ind_A$, where $A\in\F_{t-1}$. This gives $\E_R[\ind_A\Delta X_t^i]=0$ for every such $A$, hence $\E_R[\Delta X_t^i\mid\F_{t-1}]=0$. Integrability holds because the space and horizon are finite.
\end{proof}

\begin{theorem}[Completeness for a fixed projection]\label{thm:fixed}
Under the assumptions of Section~\ref{sec:setup}, the following are equivalent:
\begin{enumerate}[label=(\roman*)]
\item $L^0(\G_T)\subseteq\A_X$.
\item All measures in $\N_X$ agree on $\G_T$:
\begin{equation}\label{eq:restrictedunique}
 \bigl|\{R|_{\G_T}:R\in\N_X\}\bigr|=1.
\end{equation}
\end{enumerate}
For any attainable discounted claim $h$, its initial capital and projected wealth process satisfy
\begin{equation}\label{eq:value}
 c=\E_R[h],\qquad W_t=\E_R[h\mid\F_t]
 \quad\text{for every }R\in\N_X.
\end{equation}
These quantities are independent of $R$ and of the choice of replicating strategy.
\end{theorem}
\begin{proof}
Suppose (i) holds. For $A\in\G_T$, choose a representation $\ind_A=c+g$ with $g\in K_X$. Lemma~\ref{lem:annihilator} gives $R(A)=c$ for every $R\in\N_X$. Thus all restrictions agree. More generally, the wealth process of any representation of $h$ is an $R$-martingale ending at $h$, so~\eqref{eq:value} follows. For two replicating strategies, use the same $R$ in this conditional-expectation identity to obtain equality of their wealth processes.

For the converse, suppose that a nonnegative $\G_T$-measurable discounted claim $h$ is not in $\A_X$. On $\R^\Omega$, use the inner product $\langle f,g\rangle=\E_{Q_0}[fg]$. This is positive definite because $Q_0$ has full support. The subspace $\A_X$ is closed. Let $a$ be the orthogonal projection of $h$ onto $\A_X$, and set $Z=h-a$. Then
\begin{equation}\label{eq:orthogonal}
 Z\ne0,\qquad \E_{Q_0}[ZY]=0\quad(Y\in\A_X),\qquad
 \E_{Q_0}[Zh]=\E_{Q_0}[Z^2]>0.
\end{equation}
In particular, $\E_{Q_0}[Z]=0$ because $\ind\in\A_X$. Define a probability measure $R$ by
\begin{equation}\label{eq:perturbation}
 \frac{dR}{dQ_0}=1+\varepsilon Z,\qquad
 \varepsilon=\frac{1}{2\|Z\|_\infty}.
\end{equation}
The density lies between $1/2$ and $3/2$ and has expectation one. Hence $R\sim Q_0\sim P$. For $g\in K_X\subseteq\A_X$, equations~\eqref{eq:orthogonal} and Lemma~\ref{lem:annihilator} imply
\[
 \E_R[g]=\E_{Q_0}[g]+\varepsilon\E_{Q_0}[Zg]=0.
\]
The same lemma yields $R\in\N_X$. On the other hand,
\begin{equation}\label{eq:differentclaim}
 \E_R[h]-\E_{Q_0}[h]
 =\varepsilon\E_{Q_0}[Z^2]>0.
\end{equation}
Since $h$ is $\G_T$-measurable, $R$ and $Q_0$ cannot agree on $\G_T$. This contradicts (ii).
\end{proof}

The perturbation is the classical finite-dimensional argument of~\cite[Theorem~5.1]{BR2017}. Choosing the residual of the particular unattainable claim, rather than an arbitrary nonzero element of $\A_X^\perp$, is essential here: it proves that the perturbed measure changes the restriction to $\G_T$. A nonzero perturbation on the larger sigma-field $\F$ alone would not establish that conclusion. Neither $Z$ nor $a$ needs to be $\G_T$-measurable.

When $\mathbb F=\mathbb G$ and $\F=\G_T$, Theorem~\ref{thm:fixed} is the usual second fundamental theorem. Adaptedness alone does not imply these equalities. Additional information in the trading filtration or the ambient sigma-field must still be accounted for.

\subsection{A common projection for the original market}
Suppose that the projection is common to all original optional martingale measures:
\begin{equation}\label{eq:invariance}
 X_t^Q=X_t^{Q_0}=X_t
 \quad\text{for every }Q\in\M_S\text{ and every }t.
\end{equation}
It removes the dependence of $X$ and $\mathbb G$ on the selected original optional martingale measure. It also implies
\begin{equation}\label{eq:inclusion}
 \M_S\subseteq\N_X,
\end{equation}
since each $Q\in\M_S$ makes its own projection, now equal to $X$, a martingale.

\begin{corollary}\label{cor:necessary}
Assume~\eqref{eq:invariance}. If the projected market is complete for $\G_T$-claims, then
\begin{equation}\label{eq:originalunique}
 \bigl|\{Q|_{\G_T}:Q\in\M_S\}\bigr|=1.
\end{equation}
For every attainable claim $\xi$, its projected value is
$B_t\E_Q[\xi/B_T\mid\F_t]$ for every $Q\in\M_S$.
\end{corollary}
\begin{proof}
Both statements follow from Theorem~\ref{thm:fixed} and~\eqref{eq:inclusion}.
\end{proof}

Restricted uniqueness within $\M_S$ is therefore necessary for projected completeness. The converse would require control over measures in $\N_X$ beyond the original class $\M_S$. For example, equality of the restriction sets in~\eqref{eq:restrictedunique} and~\eqref{eq:originalunique} would suffice. Assumption~\eqref{eq:invariance} supplies only the inclusion, and does not supply that equality.

\section{Limits of the original measure criterion}\label{sec:limits}
\subsection{Uniqueness without projected completeness}
The next example shows that the missing converse cannot be recovered from projection invariance alone.

\begin{example}\label{ex:counter}
Let $\Omega=\{\omega_1,\omega_2,\omega_3\}$, $T=2$, $B_t=1$, and
\[
 \F_0=\F_1=\{\varnothing,\Omega\},\qquad \F_2=2^\Omega.
\]
For one risky asset, take the following prices.
\begin{center}
\begin{tabular}{c@{\qquad}rrr}
\toprule
State & $S_0$ & $S_1$ & $S_2$\\
\midrule
$\omega_1$ & 2 & 1 & 3\\
$\omega_2$ & 2 & 2 & 1\\
$\omega_3$ & 2 & 3 & 2\\
\bottomrule
\end{tabular}
\end{center}
Writing $q_i=Q(\{\omega_i\})$, condition~\eqref{eq:momentcondition} becomes
\begin{equation}\label{eq:counterconstraints}
 q_1+q_2+q_3=1,\qquad
 q_1+2q_2+3q_3=2,\qquad
 3q_1+q_2+2q_3=2.
\end{equation}
Subtracting twice the first equation from the second gives $q_3=q_1$; doing the same with the third gives $q_1=q_2$. Thus $\M_S=\{Q_0\}$, where $Q_0$ is uniform. The model satisfies~\eqref{eq:firstftap}, and~\eqref{eq:invariance} holds automatically.

Nevertheless, its projected process is
\[
 X_0=X_1=2,\qquad X_2=S_2.
\]
The three terminal values of $X$ are distinct, so $\G_T=2^\Omega$. Every predictable strategy has constant $H_1,H_2$, and therefore
\[
 \A_X=\{c+h(S_2-2):c,h\in\R\}.
\]
The claim $\ind_{\{\omega_2\}}$ is not in this space. Evaluation at $\omega_3$ forces $c=0$, and evaluation at $\omega_1$ then forces $h=0$, contradicting its value at $\omega_2$. The projected market is incomplete, despite uniqueness of the original optional martingale measure on the entire sigma-field.

The additional measures required by Theorem~\ref{thm:fixed} are visible explicitly:
\begin{equation}\label{eq:Ncounter}
 \N_X=\{(a,a,1-2a):0<a<1/2\}.
\end{equation}
For example, $R=(1/4,1/4,1/2)$ makes the fixed $X$ a martingale, but $\E_R[S_1]=9/4\ne2$, so $R\notin\M_S$.
\end{example}

A perturbation that annihilates gains of $X^{Q_0}$ belongs to $\N_X$, but need not belong to $\M_S$. Applying projection invariance to that measure before establishing its membership in $\M_S$ would be circular. Example~\ref{ex:counter} makes the obstruction explicit: a measure can preserve the martingale property of $X^{Q_0}$ while violating an original price-increment constraint.

\subsection{Uniqueness is relative to the claim sigma-field}\label{sec:claimfield}
A one-period market with trivial information illustrates the role of the claim sigma-field. Let $S_0=100$, let $S_1$ take the values $120,100,80$, let $B=1$, and set $\F_0=\F_1=\{\varnothing,\Omega\}$. Then
\[
 \M_S=\{(a,1-2a,a):0<a<1/2\}.
\]
For every such measure, $X_0=X_1=100$, and $\G_T$ is trivial. Every $\G_T$-claim is a constant and is attainable using cash alone. All restrictions to $\G_T$ coincide, although the measures on $2^\Omega$ differ. Completeness for a restricted claim class does not require uniqueness on unobserved events.

\subsection{Dependence on the projection measure}\label{sec:dependence}
Projection invariance is itself a substantive assumption. In the following four-state example, all listed probabilities are terminal-state masses. Let $T=2$, $B=1$, $S_0=100$, and let $\F_0=\F_1$ be trivial. At time two, the only observed event is $A=\{u,m_a,m_b\}$, so $\F_2=\sigma(A)$. Define
\begin{center}
\begin{tabular}{c@{\qquad}rrrr}
\toprule
State & $S_1$ & $S_2$ & $Q_1$ & $Q_2$\\
\midrule
$u$   & 120 & 110 & $1/3$ & $1/4$\\
$m_a$ & 100 & 110 & $1/6$ & $1/4$\\
$m_b$ & 100 & 90  & $1/6$ & $1/4$\\
$d$   & 80  & 90  & $1/3$ & $1/4$\\
\bottomrule
\end{tabular}
\end{center}
Under both measures, $\E[S_1]=\E[S_2]=100$, hence $Q_1,Q_2\in\M_S$ by~\eqref{eq:momentcondition}. Both projections are $100$ at times zero and one. At time two, however,
\begin{equation}\label{eq:measuredependence}
 X_2^{Q_1}=105\ind_A+90\ind_{A^c},\qquad
 X_2^{Q_2}=\frac{310}{3}\ind_A+90\ind_{A^c}.
\end{equation}
For example, the first conditional value is
$(110/3+110/6+90/6)/(2/3)=105$.
Each fixed projected market has two terminal values and is complete for $\sigma(A)$: a claim taking values $h_A,h_{A^c}$ is represented by
\[
 H_2=\frac{h_A-h_{A^c}}{x_A-90},\qquad
 c=h_{A^c}+10H_2,
\]
where $x_A$ is the corresponding value in~\eqref{eq:measuredependence} and $H_1=0$. Nevertheless, the initial values of $\ind_A$ are $Q_1(A)=2/3$ and $Q_2(A)=3/4$. Completeness of each selected projection does not define a common pricing model.

\section{The binomial market with delayed information}\label{sec:binomial}
Let $0<d<u$, let $b=1+r>0$, and assume
\begin{equation}\label{eq:binomialNA}
 d<b<u.
\end{equation}
On $\Omega=\{d,u\}^T$, with every state having positive $P$-probability, set
\begin{equation}\label{eq:binomial}
 B_t=b^t,\qquad S_t=S_0\prod_{j=1}^tY_j,\qquad
 \sbar_t=S_0\prod_{j=1}^t\frac{Y_j}{b},\qquad S_0>0,
\end{equation}
where $Y_j$ is the $j$th coordinate. The full price-history filtration is
$\mathcal D_t=\sigma(Y_1,\ldots,Y_t)=\sigma(S_0,\ldots,S_t)$.
For $k\in\mathbb N_0$, the delayed trading filtration is
\begin{equation}\label{eq:delay}
 \F_t^k=\mathcal D_{(t-k)^+},\qquad (t-k)^+=\max(t-k,0).
\end{equation}
Specify $Q_0$ as the \emph{product} measure under which the coordinates are independent and
\begin{equation}\label{eq:q}
 Q_0(Y_j=u)=q=\frac{b-d}{u-d},\qquad Q_0(Y_j=d)=1-q.
\end{equation}
The marginal probabilities alone would not specify the measure. Condition~\eqref{eq:binomialNA} gives $0<q<1$, so $Q_0\sim P$. Independence under $Q_0$, rather than under an unspecified physical measure, is used below.

\begin{proposition}[Delayed binomial projection]\label{prop:binomial}
In the model~\eqref{eq:binomial}--\eqref{eq:q}, $Q_0\in\M_S$ for the filtration $\mathbb F^k$, and
\begin{equation}\label{eq:binomialprojection}
 X_t=\E_{Q_0}[\sbar_t\mid\F_t^k]=\sbar_{(t-k)^+}.
\end{equation}
The projected market is complete for $\F_T^k$-claims. Its natural filtration equals $\mathbb F^k$, and its nonconstant evolution is a binomial market with $n=(T-k)^+$ periods.
\end{proposition}
\begin{proof}
Since $qu+(1-q)d=b$, the discounted price $\sbar$ is a martingale in $\mathbb D$ under $Q_0$. Conditioning at time $(t-k)^+$ proves~\eqref{eq:binomialprojection}. For $t\leq k$, $X_t=S_0$. For $t>k$, it reveals the successive discounted prices $\sbar_1,\ldots,\sbar_{t-k}$. Successive price ratios determine the coordinates $Y_j$, because $S_0>0$ and $d\ne u$. Hence $\sigma(X_0,\ldots,X_t)=\F_t^k$.

For $t<k$, $X_{t+1}=X_t$. For $t\geq k$,
\[
 \E_{Q_0}[X_{t+1}\mid\F_t^k]
 =\E_{Q_0}[\sbar_{t+1-k}\mid\mathcal D_{t-k}]
 =\sbar_{t-k}=X_t.
\]
Thus $X$ is an $(\mathbb F^k,Q_0)$-martingale, which proves $Q_0\in\M_S$ and absence of arbitrage in the original gain model.

For completeness, take a nonnegative $\mathcal D_n$-measurable claim $\xi$, and write $h=\xi/B_T$. If $n=0$, $h$ is constant and cash suffices. Suppose $n>0$. On the tree of histories of length at most $n$, set $v_n=h$ and recursively define
\begin{equation}\label{eq:backward}
 v_{j-1}(z)=qv_j(z,u)+(1-q)v_j(z,d),\qquad 1\leq j\leq n.
\end{equation}
Here $z$ is a history of length $j-1$, and $v_j(z,u)$ and $v_j(z,d)$ are the two successor values. If $x=\sbar_{j-1}(z)$, define
\begin{equation}\label{eq:delta}
 a_j(z)=\frac{v_j(z,u)-v_j(z,d)}{x(u-d)/b}.
\end{equation}
The denominator is positive. Equations~\eqref{eq:q} and~\eqref{eq:backward} imply, for either successor $y\in\{d,u\}$,
\begin{equation}\label{eq:localreplication}
 v_j(z,y)=v_{j-1}(z)+a_j(z)\bigl(xy/b-x\bigr).
\end{equation}
Thus $a_j$ is $\mathcal D_{j-1}$-measurable and
$h=v_0+\sum_{j=1}^n a_j\Delta\sbar_j$ on every path.

Set $c=v_0$ and use the delayed strategy
\begin{equation}\label{eq:shiftedstrategy}
 H_t=\begin{cases}
 0,&t\leq k,\\
 a_{t-k},&t>k,
 \end{cases}\qquad 1\leq t\leq T.
\end{equation}
For $t>k$, its measurability is exactly $\mathcal D_{t-k-1}=\F_{t-1}^k$; before then it is deterministic. By~\eqref{eq:binomialprojection},
\[
 c+(H\cdot X)_T
 =v_0+\sum_{j=1}^na_j\Delta\sbar_j=h.
\]
The cash holding~\eqref{eq:cash} makes this strategy self-financing at the projected prices. Multiplication by $B_T$ yields $\xi$. In particular, discounting uses the actual maturity $T$, even though the effective binomial horizon is $n$.
\end{proof}

The effective binomial replication is the familiar construction of Cox, Ross and Rubinstein~\cite{CRR1979}, applied to discounted prices and shifted in trading time. By Theorem~\ref{thm:fixed}, the measures in $\N_X$ agree on $\F_T^k$. Their extensions to the full path sigma-field need not be unique.

\subsection{A two-period claim}\label{sec:binomialexample}
Take $T=2$, $k=1$, $S_0=100$, $u=1.2$, $d=0.8$, and $b=1$. The measure $Q_0$ is uniform on the four paths. Then
\[
 X_0=X_1=100,\qquad X_2=S_1,\qquad
 C=\ind_{\{S_1=120\}}.
\]
A projected hedge holds $H_1=H_2=1/40$ shares and $\beta_1=\beta_2=-2$ units of cash, with initial capital $1/2$. Since the projected price has no first-period movement, using $H_1=0$, $\beta_1=1/2$ as in~\eqref{eq:shiftedstrategy} gives the same projected wealth. For the constant-holdings version, terminal values are as follows.
\begin{center}
\begin{tabular}{c@{\qquad}rrrr}
\toprule
Path & $S_2$ & $X_2$ & $-2+X_2/40=C$ & $-2+S_2/40$\\
\midrule
$uu$ & 144 & 120 & 1 & $8/5$\\
$ud$ & 96  & 120 & 1 & $2/5$\\
$du$ & 96  & 80  & 0 & $2/5$\\
$dd$ & 64  & 80  & 0 & $-2/5$\\
\bottomrule
\end{tabular}
\end{center}
The strategy is self-financing in both markets because its holdings are constant. It replicates $C$ at the projected prices, but its terminal value at the original prices differs from $C$. In this particular example that terminal value has conditional expectation $C$ under $Q_0$ given $\F_2^1$. This equality for one strategy does not imply~\eqref{eq:notcommute} in general.

Indeed, take the equally predictable risky strategy $H_1=1/40$, $H_2=0$. At the original prices,
\begin{equation}\label{eq:gainreplication}
 \frac12+(H\cdot S)_2=\frac12+\frac{S_1-100}{40}=C,
\end{equation}
whereas $\frac12+(H\cdot X)_2=\frac12$. The left-hand side of~\eqref{eq:gainreplication} is already $\F_2^1$-measurable, disproving the general conditional-projection identity. Its original cash position in the second period is $C$, which is not $\F_1^1$-measurable. It is admissible under the risky-strategy gain convention~\eqref{eq:originalgains}, but not under a convention requiring both holdings to be predictable.

Under the stricter convention requiring both holdings to be predictable, $\beta_2$ and $H_2$ are constants. Since $S_2(ud)=S_2(du)=96$ but $C(ud)\ne C(du)$, no such terminal portfolio can replicate $C$. That conclusion must not be attributed to the different gain model~\eqref{eq:originalgains}.

\subsection{Delay does not imply projection invariance}\label{sec:delayinvariance}
Proposition~\ref{prop:binomial} uses the particular product measure $Q_0$. Even in its two-period example, other original optional martingale measures can give different projections. In path order $(uu,ud,du,dd)$, let
\[
 Q_1=(1/3,1/6,1/8,3/8).
\]
The first-step up probability is $1/2$, and direct calculation gives
$\E_{Q_1}[S_1]=\E_{Q_1}[S_2]=100$.
Since $\F_0^1=\F_1^1$ is trivial, $Q_1\in\M_S$. However,
\begin{equation}\label{eq:binomialdependence}
 \E_{Q_1}[S_2\mid S_1=120]=128,\qquad
 \E_{Q_1}[S_2\mid S_1=80]=72.
\end{equation}
The projection therefore differs from $X_2=S_1$ under $Q_0$. The hypothesis that every original optional martingale measure also be a martingale measure in the full price filtration would justify the common delayed-price formula. It is not a consequence of delayed information and fails here.

\section{Scope of the completeness statement}
There are three distinct choices in the completeness question: the process used for gains, the filtration used for trading, and the terminal sigma-field of claims. For a fixed optional projection, finite-dimensional duality gives Theorem~\ref{thm:fixed} with the full class $\N_X$ of martingale measures of that process. Under a common-projection assumption, completeness also gives restricted uniqueness within $\M_S$, as in Corollary~\ref{cor:necessary}. Example~\ref{ex:counter} shows why the latter class cannot generally replace the former in an equivalence.

The delayed binomial model gives a constructive positive result at a specified measure. It preserves all claims observable in the projected price history while reducing the effective trading horizon. The result identifies the completeness of an auxiliary adapted market. Replication at the original prices requires a separate check of both terminal payoff and the admissibility of the cash and risky positions.

\appendix
\input{elementary}
\input{first_ftap}
\input{supplement}
\input{computation}

\clearpage

\input{main.bbl}
\end{document}

%% file: elementary.tex
\section{Supporting martingale and portfolio arguments}\label{app:elementary}
The auxiliary results in this appendix are classical; the financial statements follow the framework of~\cite{BR2017}. A probability space in this appendix need not be finite unless explicitly stated.

\subsection{Measurability and optional projection}
\begin{lemma}\label{lem:trivial}
Every real-vector-valued function measurable with respect to $\{\varnothing,\Omega\}$ on a nonempty space is constant.
\end{lemma}
\begin{proof}
Choose $y$ in the range of the function $f$. The singleton $\{y\}$ is Borel, so its inverse image is measurable and nonempty. It must therefore equal $\Omega$, which gives $f=y$ everywhere.
\end{proof}

\begin{lemma}\label{lem:onestep}
An adapted integrable process $Y=(Y_t)_{t=0}^T$ is a martingale if and only if $\E[Y_t\mid\F_{t-1}]=Y_{t-1}$ for $1\leq t\leq T$.
\end{lemma}
\begin{proof}
Necessity follows from the definition. Conversely, for $s<t$, conditional linearity and the tower property give
\[
\E[Y_t\mid\F_s]
=Y_s+\sum_{j=s+1}^t\E\bigl[\E[Y_j-Y_{j-1}\mid\F_{j-1}]\mid\F_s\bigr]
=Y_s.
\]
Every conditional expectation exists by integrability.
\end{proof}

\begin{lemma}\label{lem:optionalstopping}
For an integrable process $Y$ at finitely many dates, $Z_t=\E[Y_t\mid\F_t]$ is its discrete optional projection: for any stopping time $\tau$ with values in $\{0,\ldots,T,\infty\}$,
\[
\E\left[\sum_{t=0}^T Y_t\ind_{\{\tau=t\}}\,\middle|\,\F_\tau\right]
=\sum_{t=0}^T Z_t\ind_{\{\tau=t\}}.
\]
\end{lemma}
\begin{proof}
For $A\in\F_\tau$, each $A\cap\{\tau=t\}$ belongs to $\F_t$. Thus
\[
\E\left[\ind_A\sum_{t=0}^T Y_t\ind_{\{\tau=t\}}\right]
=\sum_{t=0}^T\E[Y_t\ind_{A\cap\{\tau=t\}}]
=\sum_{t=0}^T\E[Z_t\ind_{A\cap\{\tau=t\}}].
\]
The sum on the right in the asserted identity is $\F_\tau$-measurable and integrable. This proves the conditional-expectation identity. Uniqueness follows by choosing each deterministic $\tau=t$. This argument uses the decomposition over finitely many dates. The event $\{\tau<\infty\}$ is also $\F_\tau$-measurable in continuous time; its measurability is not a distinction between the two settings.
\end{proof}

\begin{proposition}\label{prop:integration}
Let $Y$ be an adapted real process with $Y_0\in L^1$. Then $Y$ is a martingale if and only if $H\cdot Y$ is a martingale for every bounded predictable $H$.
\end{proposition}
\begin{proof}
If $Y$ is a martingale, boundedness of $H$ and the finite horizon give integrability of its stochastic integral. Adaptedness and predictability yield
\[
\E[(H\cdot Y)_t\mid\F_{t-1}]
=(H\cdot Y)_{t-1}+H_t\E[Y_t-Y_{t-1}\mid\F_{t-1}]
=(H\cdot Y)_{t-1}.
\]
Lemma~\ref{lem:onestep} applies. Conversely, take $H\equiv1$. Since $(H\cdot Y)_t=Y_t-Y_0$ is integrable, so is $Y_t$. For a fixed $j\geq1$, take $H_t=\ind_{\{t=j\}}$. Its integral is zero at $j-1$ and equals $Y_j-Y_{j-1}$ at $j$. The martingale property gives $\E[Y_j-Y_{j-1}\mid\F_{j-1}]=0$. Apply Lemma~\ref{lem:onestep} again. For vector processes, the same argument is applied to each coordinate.
\end{proof}

\subsection{Self-financing identities}
Let $P_t=(B_t,S_t)$ be adapted, with deterministic $B_t>0$, $B_0=1$, and deterministic $P_0$. For predictable $\varphi_t=(\beta_t,H_t)$, define $V_t=\varphi_{t+1}\cdot P_t$ for $t<T$ and $V_T=\varphi_T\cdot P_T$.

\begin{proposition}\label{prop:portfolioidentities}
The self-financing conditions $\varphi_t\cdot P_t=\varphi_{t+1}\cdot P_t$ for $1\leq t<T$ are equivalent to
\begin{equation}\label{eq:undiscountedgains}
V_t=V_0+\sum_{s=1}^t\varphi_s\cdot(P_s-P_{s-1}),\qquad 1\leq t\leq T.
\end{equation}
They are also equivalent to the same self-financing conditions with $P_t$ replaced by $(1,\sbar_t)$. For a self-financing strategy,
\begin{equation}\label{eq:discountedgainsapp}
V_t/B_t=V_0+\sum_{s=1}^t H_s\cdot\Delta\sbar_s.
\end{equation}
Every predictable risky strategy $H$ has exactly one predictable self-financing cash completion for a specified initial capital.
\end{proposition}
\begin{proof}
For $s<T$, self-financing gives
\[
V_s-V_{s-1}=\varphi_{s+1}\cdot P_s-\varphi_s\cdot P_{s-1}
=\varphi_s\cdot(P_s-P_{s-1}).
\]
The same equality for $s=T$ follows directly from the terminal definition. Summing proves~\eqref{eq:undiscountedgains}. Conversely, subtract that identity at consecutive dates $t$ and $t-1$, for $t<T$, and substitute the definition of $V$ to obtain self-financing.

Dividing $\varphi_t\cdot P_t=\varphi_{t+1}\cdot P_t$ by the positive number $B_t$ proves the discounted equivalence. For each $s$, using the holdings carried over $(s-1,s]$ gives
\[
\frac{V_s}{B_s}-\frac{V_{s-1}}{B_{s-1}}
=\beta_s+H_s\cdot\sbar_s-\beta_s-H_s\cdot\sbar_{s-1}
=H_s\cdot\Delta\sbar_s.
\]
Summation proves~\eqref{eq:discountedgainsapp}. Finally, for initial capital $c$, necessarily
\[
\beta_1=c-H_1\cdot S_0,\qquad
\beta_{t+1}=\beta_t+(H_t-H_{t+1})\cdot\sbar_t.
\]
These equations uniquely define the cash holdings and ensure self-financing. Adaptedness of $S$ makes $\beta_{t+1}$ $\F_t$-measurable by induction. This last conclusion is precisely what needs rechecking if $S$ is nonadapted.
\end{proof}

An immediate consequence is the value-process characterization of optional martingale measures. On the finite space of Section~\ref{sec:setup}, fix $Q\sim P$ and form its own projection $X^Q$. Then $Q\in\M_S$ if and only if every discounted self-financing wealth process in the market with risky prices $B_tX_t^Q$ is a $Q$-martingale. Indeed, Proposition~\ref{prop:portfolioidentities} writes every such process as $c+H\cdot X^Q$. Proposition~\ref{prop:integration} proves one direction. For the other, any predictable $H$ has the cash completion just proved, so the assumed wealth property applies to every stochastic integral and gives the martingale property of $X^Q$. All strategies here are bounded because $\Omega$ is finite.

\subsection{Equivalent measures and one-period arbitrage}
\begin{lemma}
Whether a fixed strategy is an arbitrage does not change when the physical probability measure is replaced by an equivalent measure. Discounting by a strictly positive numeraire also preserves the arbitrage inequalities for a zero-cost strategy.
\end{lemma}
\begin{proof}
Equivalent measures have exactly the same null sets. Thus $V_T\geq0$ almost surely and $P(V_T>0)>0$ hold for one precisely when they hold for the other; the pathwise self-financing equations and initial cost are unchanged. Also $V_T$ and $V_T/B_T$ have the same sign. For zero initial capital,~\eqref{eq:discountedgainsapp} identifies the latter with the discounted gain. No invariance of martingale properties under equivalent changes of measure is asserted.
\end{proof}

\begin{proposition}\label{prop:onestepNA}
For adapted prices and predictable self-financing trading, an arbitrage over $T$ dates exists if and only if, for some $t$, there is an $\F_{t-1}$-measurable vector $\eta$ such that $\eta\cdot\Delta\sbar_t\geq0$ almost surely and is strictly positive with positive probability.
\end{proposition}
\begin{proof}
Let $W=V/B$ be the wealth of an arbitrage, with $W_0=0$. Choose the smallest deterministic $t\in\{1,\ldots,T\}$ for which $W_t\geq0$ almost surely and $P(W_t>0)>0$. If $W_{t-1}=0$ almost surely, then $\eta=H_t$ has the required property. Otherwise minimality implies $P(W_{t-1}<0)>0$: a nonzero nonnegative $W_{t-1}$ would already have the defining property. Put
\[
A=\{W_{t-1}<0\}\in\F_{t-1},\qquad \eta=\ind_AH_t.
\]
Then $\eta\cdot\Delta\sbar_t=\ind_A(W_t-W_{t-1})\geq0$, and it is strictly positive on $A$.

Conversely, hold only the risky position $\eta$ during $(t-1,t]$ and no risky asset at the other dates. Start from zero capital and complete the strategy by Proposition~\ref{prop:portfolioidentities}. Its discounted terminal wealth is exactly $\eta\cdot\Delta\sbar_t$. The cash position after date $t$ retains this gain until maturity. The required measurability follows from adaptedness. This proves both implications, including the financing step.
\end{proof}

%% file: first_ftap.tex
\section{The first fundamental theorem and its auxiliary proofs}\label{app:first}
The first fundamental theorem of Kabanov and Stricker~\cite{KS2006} follows from the selection and separation arguments of~\cite{KS2001}. Unlike the completeness results in the main text, the probability space here is arbitrary. The induction and integrability steps are written out explicitly.

Fix a probability space $(\Omega,\F,P)$, a finite horizon $T$, a filtration $\mathbb F$ with $\F_T\subseteq\F$, and $\F$-measurable $\R^d$-valued random vectors $U_t$ that are finite almost surely, not necessarily adapted. In a financial application, $U=\sbar$. Put
\[
\mathcal R=\left\{\sum_{t=1}^T H_t\cdot\Delta U_t:H_t\text{ is }\F_{t-1}\text{-measurable}\right\},
\qquad \mathcal C=\mathcal R-L^0_+.
\]
Positions are finite random vectors; they need not be bounded. Closures in $L^0$ refer to convergence in probability.

\begin{theorem}[Kabanov and Stricker]\label{thm:KSgeneral}
The following conditions are equivalent:
\begin{enumerate}[label=(\roman*)]
\item $\mathcal C\cap L^0_+=\{0\}$.
\item $\mathcal C\cap L^0_+=\{0\}$ and $\mathcal C$ is closed in probability.
\item $\overline{\mathcal C}\cap L^0_+=\{0\}$.
\item There is $Q\sim P$ with $dQ/dP\in L^\infty(P)$, with every $U_t\in L^1(Q)$, such that $\E_Q[\Delta U_t\mid\F_{t-1}]=0$ for all $t$.
\end{enumerate}
\end{theorem}
The last condition is equivalent to the martingale property of $\E_Q[U_t\mid\F_t]$ by the tower property. Also $\mathcal C\cap L^0_+=\{0\}$ is equivalent to $\mathcal R\cap L^0_+=\{0\}$: if $0\leq f=g-r$ with $g\in\mathcal R$ and $r\geq0$, then $g=f+r\geq0$, so absence of nonzero positive gains forces both $g$ and $f$ to vanish.

\subsection{A measurable subsequence}
\begin{lemma}\label{lem:selection}
Let $\eta^n$ be $\mathcal A$-measurable $\R^d$-valued random vectors with $\liminf_n|\eta^n|<\infty$ almost surely. There are increasing finite $\mathcal A$-measurable integer-valued indices $\tau_j$ such that $\eta^{\tau_j}$ converges almost surely.
\end{lemma}
\begin{proof}
Work outside the null set where the liminf is infinite and define arbitrary values there. Let $a=\liminf_n|\eta^n|$. Starting from $\tau_0=0$, choose recursively the first $n>\tau_{j-1}$ for which $\bigl||\eta^n|-a\bigr|<1/j$. Such indices exist by the definition of liminf; they are measurable because this is a countable first-index selection from measurable events. The selected vector sequence is bounded pointwise.

For its first coordinate, let $a_1$ be the pointwise liminf and repeat the first-index selection with distance less than $1/j$ from $a_1$. This produces a subsequence whose first coordinate converges. Repeat for coordinates $2,\ldots,d$, each time taking a subsequence of the preceding sequence. Coordinatewise limits from earlier steps are preserved. After $d$ steps the vectors converge; compositions of the measurable indices are measurable. These indices need not be stopping times for the trading filtration.
\end{proof}

\subsection{Closedness of the attainable cone}
\begin{proof}[Proof of Theorem~\ref{thm:KSgeneral}, (i)$\Rightarrow$(ii)]
The proof is by induction on $T$, with the zero-period cone $-L^0_+$ as base. Write $D_t=\Delta U_t$. A sequence converging in probability has an almost surely convergent subsequence, so suppose
\[
f_n=\sum_{t=1}^T H_t^n\cdot D_t-r_n\longrightarrow f\quad\text{a.s.},\qquad r_n\geq0.
\]
The cone generated by dates $2,\ldots,T$, with the same information at those dates, also satisfies (i). By induction it is closed. Localization on an event in $\F_0$ preserves (i) and predictability, since a strategy on that event extends by zero outside it.

First work on the $\F_0$-measurable event where $\liminf_n|H_1^n|<\infty$. Lemma~\ref{lem:selection}, applied with $\mathcal A=\F_0$, allows selection of a random subsequence along which $H_1^n\to H_1$. Selecting every $H_t^n$ and $r_n$ with the same indices preserves predictability. Then
\[
\sum_{t=2}^T H_t^n\cdot D_t-r_n\longrightarrow f-H_1\cdot D_1.
\]
Closedness of the remaining-date cone gives a representation of this limit, and hence one of $f$ in $\mathcal C$.

On the complementary event $|H_1^n|\to\infty$, discard any initial zero denominators by a measurable selection and divide all positions and $r_n$ by $|H_1^n|$. Choose a further $\F_0$-measurable subsequence with
\[
G_1^n:=H_1^n/|H_1^n|\longrightarrow G_1,\qquad |G_1|=1.
\]
Since $f_n/|H_1^n|\to0$ almost surely, the normalized cone elements satisfy
\[
\sum_{t=2}^T \frac{H_t^n}{|H_1^n|}\cdot D_t-\frac{r_n}{|H_1^n|}
\longrightarrow -G_1\cdot D_1.
\]
The induction hypothesis supplies predictable $G_2,\ldots,G_T$ and $r\geq0$ such that
\[
\sum_{t=2}^T G_t\cdot D_t-r=-G_1\cdot D_1.
\]
Thus $\sum_{t=1}^T G_t\cdot D_t=r\geq0$. Condition (i) implies that this gain and $r$ are zero.

Partition the event into the $d$ measurable sets on which a specified coordinate is the first nonzero coordinate of $G_1$. On the piece with index $i$, put $b_n=H_1^{n,i}/G_1^i$ and replace all positions by
\[
\widehat H_t^n=H_t^n-b_nG_t,\qquad 1\leq t\leq T.
\]
The multipliers are $\F_0$-measurable, so predictability is preserved. The terminal gain is unchanged and $\widehat H_1^{n,i}=0$. Any coordinates of $H_1^n$ that had already been eliminated stay zero, because the corresponding coordinates of $G_1$ are also zero.

Repeat the bounded-subsequence or elimination argument on each piece. Every elimination removes one coordinate of the first position. After at most $d$ eliminations its sequence has a finite liminf and the first part of the proof applies. The partition has only finitely many pieces at each stage. Patching their predictable strategies and nonnegative remainders yields $f\in\mathcal C$. This proves closedness and completes the induction.
\end{proof}

\subsection{Separation and a strictly positive density}
\begin{lemma}[Kreps--Yan separation in $L^1$]\label{lem:separation}
Let $K$ be a norm-closed convex cone in real $L^1(P)$, containing $-L^1_+$, with $K\cap L^1_+=\{0\}$. There is a bounded density $\rho>0$ almost surely, with $\E_P\rho=1$, such that $\E_P[\rho f]\leq0$ for every $f\in K$.
\end{lemma}
\begin{proof}
The proof is the separation argument used in~\cite[Lemma~3]{KS2001}. For each nonzero $x\in L^1_+$, separate $x$ strictly from the closed convex set $K$. The Hahn--Banach separation theorem and $(L^1)'=L^\infty$ give $z_x\in L^\infty$ with
\[
\E[z_xx]>0,\qquad \E[z_xf]\leq0\quad(f\in K).
\]
Indeed, the functional is bounded above on $K$; since $K$ is a cone containing zero, that upper bound forces its value to be nonpositive there. Since $-L^1_+\subseteq K$, testing $-\ind_{\{z_x<0\}}$ shows $z_x\geq0$. Rescale to $0\leq z_x\leq1$.

A countable support selection gives a strictly positive separator. Let $a$ be the supremum of $P(\bigcup_{j=1}^m\{z_{x_j}>0\})$ over finite subfamilies. Choose finite subfamilies with union probabilities tending to $a$, and enumerate their combined, countable family as $(z_j)$. Its support union $A$ has probability $a$: finite partial unions have probability at most $a$, while the selected unions approach $a$. For any $x$, the same finite-union bound, followed by continuity from below, gives $P(A\cup\{z_x>0\})\leq a$. Thus $z_x=0$ almost surely on $A^c$. If $P(A^c)>0$, choosing $x=\ind_{A^c}$ contradicts $\E[z_xx]>0$. Hence $P(A)=1$.

Set $\rho_0=\sum_{j\geq1}2^{-j}z_j$. Then $0<\rho_0\leq1$ almost surely. For $f\in K$, dominated convergence, with dominating variable $|f|$, gives
\[
\E[\rho_0 f]=\lim_{m\to\infty}\sum_{j=1}^m2^{-j}\E[z_jf]\leq0.
\]
Normalize by $\rho=\rho_0/\E\rho_0$. It is bounded, has expectation one, is strictly positive almost surely, and has the desired separating property. Dominated convergence is used because $f$ can have both signs.
\end{proof}

\begin{proof}[Proof of Theorem~\ref{thm:KSgeneral}, (ii)$\Rightarrow$(iii)$\Rightarrow$(iv)]
The first implication is immediate. For the second, first make the prices integrable. With $Z=\sum_{t=0}^T|U_t|$, let
\[
\frac{dP^*}{dP}=\frac{e^{-Z}}{\E_P[e^{-Z}]}.
\]
This density is bounded and strictly positive, and $Ze^{-Z}$ is bounded, so every $U_t$ is $P^*$-integrable. Equivalent probabilities induce the same convergence-in-probability topology: if $P(A_n)\to0$, integrability of $dP^*/dP$ gives $P^*(A_n)\to0$, and the reverse follows by equivalence. Consequently (iii), which concerns null sets and closure in that topology, is unchanged.

In $L^1(P^*)$, set $K=\overline{\mathcal C}\cap L^1(P^*)$. It is a convex cone containing $-L^1_+(P^*)$ and has trivial intersection with $L^1_+(P^*)$. It is norm-closed because $L^1$ convergence implies convergence in probability. Lemma~\ref{lem:separation} gives $Q\sim P^*$ with bounded $dQ/dP^*$ and $\E_Q f\leq0$ for $f\in K$. For bounded $\F_{t-1}$-measurable $H_t$, both signs of $H_t\cdot\Delta U_t$ belong to $K$. Therefore their $Q$-expectations are zero. Testing coordinate vectors times indicators gives $\E_Q[\Delta U_t\mid\F_{t-1}]=0$. The product
\[
\frac{dQ}{dP}=\frac{dQ}{dP^*}\frac{dP^*}{dP}
\]
is bounded, and boundedness of $dQ/dP^*$ preserves integrability of every $U_t$. This proves (iv).
\end{proof}

\subsection{Unbounded predictable positions}
The last implication requires care because integrable price increments multiplied by unbounded positions need not be integrable. The following conditioning argument supplies the missing justification without assuming adapted prices.

\begin{lemma}\label{lem:unbounded}
Suppose $D_t\in L^1(Q;\R^d)$ and $\E_Q[D_t\mid\F_{t-1}]=0$. If finite $\F_{t-1}$-measurable positions $H_t$ satisfy $Z=\sum_{t=1}^T H_t\cdot D_t\geq0$, then $Z=0$ almost surely.
\end{lemma}
\begin{proof}
Use induction on $T$. Since all $H_t$ are $\F_{T-1}$-measurable, the conditional expectation of the nonnegative $Z$ is well defined and
\begin{equation}\label{eq:localizedcondition}
\E_Q[Z\mid\F_{T-1}]
=\sum_{t=1}^{T-1}H_t\cdot\E_Q[D_t\mid\F_{T-1}].
\end{equation}
To justify this identity, restrict first to the events $A_m=\{\max_{1\leq t\leq T}|H_t|\leq m\}\in\F_{T-1}$. The variable $\ind_{A_m}Z$ is integrable, all products can be conditioned, and the term for $t=T$ vanishes. These increasing events cover $\Omega$ up to a null set, so the localized identities imply~\eqref{eq:localizedcondition}; in particular its right-hand side is finite almost surely.

For $T=1$ the sum is empty and the conditional expectation is zero, forcing $Z=0$. For $T>1$, set $D'_t=\E_Q[D_t\mid\F_{T-1}]$, $1\leq t<T$. These variables are integrable and satisfy $\E_Q[D'_t\mid\F_{t-1}]=0$ by the tower property. The right-hand side of~\eqref{eq:localizedcondition} is nonnegative. The induction hypothesis applied to $(D'_t)_{t<T}$ makes it zero. Thus $\E_Q Z=0$ in the nonnegative extended sense, and $Z=0$ almost surely.
\end{proof}

\begin{proof}[Proof of Theorem~\ref{thm:KSgeneral}, (iv)$\Rightarrow$(i)]
If $0\leq f=\sum_tH_t\cdot\Delta U_t-r$ with $r\geq0$, the sum is nonnegative. Lemma~\ref{lem:unbounded} makes it zero $Q$-almost surely, hence $P$-almost surely. It follows that $f=0$. This finishes the proof of all four equivalences.
\end{proof}

\subsection{The algebraic extension argument}\label{app:HB}
The classical real Hahn--Banach extension theorem follows from Zorn's lemma. The separation and $L^p$ duality results invoked above are the standard consequences and companion results recorded in~\cite{Rudin1991}.

\begin{theorem}[Real Hahn--Banach extension]
Let $M$ be a linear subspace of a real vector space $E$, let $p:E\to\R$ be sublinear, and let $f:M\to\R$ be linear with $f\leq p$ on $M$. There is a linear extension $\Lambda:E\to\R$ with $-p(-x)\leq\Lambda(x)\leq p(x)$ for every $x\in E$.
\end{theorem}
\begin{proof}
First choose $x_0\notin M$ and seek an extension to $M+\R x_0$ of the form $\Lambda_r(m+\lambda x_0)=f(m)+\lambda r$. For $\lambda=0$ domination is given. For positive $\lambda$, homogeneity reduces it to $r\leq p(v+x_0)-f(v)$ for all $v\in M$. For negative $\lambda$, it reduces to $r\geq f(w)-p(w-x_0)$ for all $w\in M$. Such an $r$ exists because
\[
f(v)+f(w)=f(v+w)\leq p(v+w)
\leq p(v+x_0)+p(w-x_0).
\]
Thus every lower bound is at most every upper bound. Both families contain finite numbers, so
\[
\sup_{w\in M}\{f(w)-p(w-x_0)\}
\leq r\leq
\inf_{v\in M}\{p(v+x_0)-f(v)\}
\]
has a real solution, giving the one-dimensional extension.

Partially order pairs $(V,g)$ consisting of a subspace $M\subseteq V\subseteq E$ and an algebraic linear extension $g$ of $f$, dominated by $p$, by extension of both domain and function. The family is nonempty because it contains $(M,f)$. A chain has upper bound with domain the union of its subspaces; the functions agree on intersections by the chain property, so they define a linear dominated function on this union. Zorn's lemma gives a maximal pair $(V,g)$. If $V\ne E$, the one-dimensional construction extends it further, a contradiction. Hence $\Lambda=g$ on $E$. Finally $\Lambda(-x)\leq p(-x)$ and linearity give the lower bound. No continuity assumption is made at this algebraic stage.
\end{proof}

%% file: supplement.tex
\section{Further projection arguments and examples}\label{app:supplement}
\subsection{Conditions for a common projection}
A singleton $\M_S$ satisfies~\eqref{eq:invariance} automatically. More generally, suppose the information filtration is constant, $\F_t=\mathcal A$, and $S_0$ is deterministic. For any optional martingale measure $Q$, $X_t^Q$ is $\mathcal A$-measurable and its martingale property gives
\[
X_t^Q=\E_Q[X_t^Q\mid\mathcal A]=X_0^Q=S_0.
\]
Thus every projection is the same constant process, even when $\mathcal A$ is nontrivial. The trivial-information trinomial model of Section~\ref{sec:claimfield} is a special case.

For a delayed filtration $\F_t=\mathcal D_{(t-k)^+}$, where $\mathbb D$ adapts discounted prices, suppose additionally that every $Q\in\M_S$ is a martingale measure for $\sbar$ in $\mathbb D$. Then
\[
X_t^Q=\E_Q[\sbar_t\mid\mathcal D_{(t-k)^+}]=\sbar_{(t-k)^+}
\]
by the martingale property. This proves invariance under that additional assumption. The example in Section~\ref{sec:delayinvariance} shows why membership in $\M_S$ alone cannot replace it. Adaptedness itself gives $X^Q=\sbar$, but does not imply that the natural price filtration is the entire trading filtration or that its terminal sigma-field equals the ambient one.

\subsection{Arbitrage-free price processes}\label{app:pricing}
The value of a projected replicating portfolio and an arbitrage-free price obtained by adding a claim as a traded asset refer to different market constructions. Their relation requires an explicit comparison of the corresponding gains and martingale measures.

\begin{proposition}\label{prop:extensionpricing}
In the finite setting of Section~\ref{sec:setup}, let $\xi\geq0$ be $\F_T$-measurable and let $Q_0\in\M_S$. The adapted process
\begin{equation}\label{eq:extensionprice}
\pi_t=B_t\E_{Q_0}[\xi/B_T\mid\F_t]
\end{equation}
is nonnegative, has $\pi_T=\xi$, and its addition to the original gain market preserves absence of arbitrage. If $\xi$ is attainable in the fixed projected market, this is the unique arbitrage-free adapted price process in that projected market and equals every replicating portfolio's value. If, in addition,~\eqref{eq:invariance} holds, it is also the unique arbitrage-free adapted price process for the claim added to the original gain market.
\end{proposition}
\begin{proof}
Conditional expectation and $B_t>0$ give nonnegativity and the terminal value. The process $\pi/B$ is an $(\mathbb F,Q_0)$-martingale. Its own optional projection is itself. Since $Q_0\in\M_S$, the vector formed by the original projected prices and $\pi/B$ is a $Q_0$-martingale. The first fundamental theorem applied to the enlarged gain space proves absence of arbitrage.

If $\xi$ is attainable in the fixed projected market, Theorem~\ref{thm:fixed} yields $V_t=B_t\E_R[\xi/B_T\mid\F_t]$ for every $R\in\N_X$ and for every replicating strategy. This proves equality of their values and independence from the chosen measure. If $\widehat\pi$ is another adapted price process with terminal value $\xi$ whose addition to that projected market is arbitrage-free, its enlarged market has a martingale measure $R\in\N_X$. Hence
\[
\widehat\pi_t/B_t=\E_R[\widehat\pi_T/B_T\mid\F_t]
=\E_R[\xi/B_T\mid\F_t]=V_t/B_t.
\]
This proves uniqueness there without transferring a projected hedge to the original asset.

Under~\eqref{eq:invariance}, every $Q\in\M_S$ belongs to $\N_X$. If addition of $\widehat\pi$ to the original market is arbitrage-free, an optional martingale measure $Q$ of the enlarged market belongs to $\M_S$. Since $\widehat\pi$ is adapted, $\widehat\pi/B$ is a $Q$-martingale. Apply the same terminal conditioning identity and Theorem~\ref{thm:fixed} to obtain $\widehat\pi=V=\pi$. Without the common-projection assumption, Section~\ref{sec:dependence} supplies distinct original-market prices $2/3$ and $3/4$ for the same claim. Thus fixed-measure projected attainability alone does not imply the last uniqueness assertion.
\end{proof}

\subsection{A two-state information asymmetry}\label{app:asymmetry}
Let $T=2$, $B\equiv1$, and give both states positive physical probability. Set
\[
\begin{array}{c|ccc}
&S_0&S_1&S_2\\\hline
u&100&120&100\\
d&100&80&100
\end{array}
\qquad
\F_0=\F_1=\{\varnothing,\Omega\},\quad\F_2=2^\Omega.
\]
An optional martingale measure must satisfy $120q+80(1-q)=100$, so it is uniquely $Q(u)=Q(d)=1/2$. The projected process is identically 100. Its natural terminal sigma-field is trivial, so the projected market is complete for its price-history claims; it is not complete for all $\F_2$-claims.

For the original prices, if both cash and risky positions must be predictable, $\beta_2$ and $H_2$ are deterministic by Lemma~\ref{lem:trivial}. Every terminal value is then the constant $\beta_2+100H_2$, so $C=\ind_{\{u\}}$ cannot be replicated. The self-financing equation at time one gives
\[
\beta_1+120H_1=\beta_2+120H_2,\qquad
\beta_1+80H_1=\beta_2+80H_2.
\]
Subtracting shows $H_1=H_2$ and then $\beta_1=\beta_2$. In contrast, under the gain convention, $c=1/2$, $H_1=1/40$, $H_2=0$ gives
\[
c+(H\cdot S)_2=\frac12+\frac{S_1-100}{40}=C.
\]
Its second cash position equals $C$ and is unobservable at time one. Nonattainability of this claim therefore depends on the stricter portfolio convention.

Now give a second investor the natural price filtration $\mathbb D$. At time one that investor can use
\[
H_1=\beta_1=0,\qquad H_2=-\ind_{\{u\}}+\ind_{\{d\}},\qquad
\beta_2=120\ind_{\{u\}}-80\ind_{\{d\}}.
\]
The strategy is $\mathbb D$-predictable. Its cost at time one is zero in both states and its terminal value is 20 in both states, so it is a self-financing arbitrage. Equivalently, no full-information martingale measure exists because $\E_Q[S_2\mid\mathcal D_1]=100\ne S_1$. The delayed investor's original gain market has the optional martingale measure above and is arbitrage-free. The example establishes a dependence on information in this model; it is not an empirical claim about high-frequency trading.

\subsection{The four-period information example}\label{app:fourperiod}
Take $T=4$, $k=1$, $S_0=1$, and write $[z]=\{\omega\in\{u,d\}^4:\omega\text{ starts with }z\}$ for a cylinder set. The delayed sigma-fields are
\[
\begin{aligned}
\F_0^1=\F_1^1&=\{\varnothing,\Omega\},\\
\F_2^1&=\sigma([u],[d]),\\
\F_3^1&=\sigma([uu],[ud],[du],[dd]),\\
\F_4^1&=\sigma([uuu],[uud],[udu],[udd],[duu],[dud],[ddu],[ddd]).
\end{aligned}
\]
Each sigma-field consists of \emph{all unions} of the displayed atoms. In particular their numbers of atoms are $1,1,2,4,8$, and their numbers of events are $2,2,4,16,256$. The price process $S$ itself is not adapted to this delayed filtration in the nondegenerate model. At time three the information is $\sigma(S_0,S_1,S_2)$, not merely $\sigma(S_2)$: the histories $ud$ and $du$ have the same current price $ud$ but different observed price histories. Under the product martingale measure the discounted projection is
\[
(X_0,X_1,X_2,X_3,X_4)
=(1,1,\sbar_1,\sbar_2,\sbar_3).
\]
It is adapted to exactly the delayed history filtration displayed above. These sigma-fields specify the history available to the investor at each date.

\subsection{The strategy-shift proof in the binomial model}\label{app:shift}
An alternative proof of Proposition~\ref{prop:binomial} uses a full-information replicating portfolio and shifts its trading times. It complements the explicit construction in~\eqref{eq:backward}--\eqref{eq:shiftedstrategy}.

\begin{proof}[Alternative proof of projected binomial completeness]
Let $n=(T-k)^+$ and $h=\xi/B_T$, with $\xi$ nonnegative and $\mathcal D_n$-measurable. If $n=0$, take $H=0$ and $\beta_t=h$ at every date. Suppose $n\geq1$. Classical binomial completeness gives a full-information self-financing portfolio $(\widehat\beta_j,\widehat H_j)_{j=1}^T$ with terminal discounted wealth $h$; this also follows from the backward construction in the main proof with horizon $T$. Under the product measure $Q_0$ its discounted wealth $\widehat W$ is a martingale. Therefore
\begin{equation}\label{eq:earlyknownclaim}
\widehat W_n=\E_{Q_0}[h\mid\mathcal D_n]=h.
\end{equation}
Set $(\widehat\beta_0,\widehat H_0)=(\widehat\beta_1,\widehat H_1)$ and define
\[
(\beta_t,H_t)=(\widehat\beta_{(t-k)^+},\widehat H_{(t-k)^+}),\qquad 1\leq t\leq T.
\]
For $t\leq k$, these are deterministic initial holdings. For $t>k$, their measurability is $\mathcal D_{t-k-1}=\F^k_{t-1}$. Thus the full portfolio is predictable.

Until the first observed movement, both $X$ and the holdings are constant. At the transition from index zero to index one the holdings agree by definition. For a later rebalancing with $j=t-k\geq1$, ordinary discounted self-financing gives
\[
\beta_t+H_tX_t
=\widehat\beta_j+\widehat H_j\sbar_j
=\widehat\beta_{j+1}+\widehat H_{j+1}\sbar_j
=\beta_{t+1}+H_{t+1}X_t.
\]
It is therefore self-financing in the projected market. Finally $X_T=\sbar_n$, so~\eqref{eq:earlyknownclaim} implies
\[
B_T(\beta_T+H_TX_T)
=B_T(\widehat\beta_n+\widehat H_n\sbar_n)
=B_T\widehat W_n=\xi.
\]
This includes $k=0$. The cash holdings are units of the bank account; the projected risky price is $B_t\sbar_{(t-k)^+}$, not the unscaled old price $S_{(t-k)^+}$ when interest is nonzero.
\end{proof}

Under the product measure, the conditional expectations can also be calculated directly. With $m=(t-k)^+$,
\[
\E_{Q_0}[\sbar_t\mid\mathcal D_m]
=\sbar_m\prod_{i=m+1}^t\E_{Q_0}[Y_i/b]
=\sbar_m.
\]
For $t<k$, $m=0$ and this is $S_0$. For $t\geq k$,
\[
\E_{Q_0}[\sbar_{t+1}\mid\mathcal D_{t-k}]
=\sbar_{t-k}\prod_{i=t-k+1}^{t+1}\E_{Q_0}[Y_i/b]
=\sbar_{t-k}
=\E_{Q_0}[\sbar_t\mid\mathcal D_{t-k}].
\]
The first product on this last line has $k+1$ factors, each equal to one. Independence is used under $Q_0$. In the nondegenerate full-support binomial model the inequality $d<b<u$ is also necessary for absence of original gain arbitrage: if $b\leq d$, a long position in the first period has a nonnegative discounted increment and a strictly positive one on an up state; if $b\geq u$, a short position has the corresponding property on a down state. These deterministic first-period positions remain available under any delay.

\subsection{Two examples behind the first fundamental theorem}
The following examples are due to Kabanov and Stricker~\cite[Sections~2 and~4]{KS2006}.

Let $T=2$, $B\equiv1$, $\F_0=\F_1$ trivial, and $0<P(A)<1$. Define
\[
\Delta S_1=\ind_A-\tfrac12\ind_{A^c},\qquad
\Delta S_2=-\tfrac12\ind_A+\ind_{A^c}.
\]
Take $S_0=1$, so the prices are positive. At either individual date every nonzero admissible position is constant and its gain takes both signs. There is no one-period arbitrage. However, $H_1=H_2=1$ gives
\[
(H\cdot S)_2=(1-\tfrac12)\ind_A+(-\tfrac12+1)\ind_{A^c}=\tfrac12.
\]
This is a two-period arbitrage and explains the failure of Proposition~\ref{prop:onestepNA} without adaptedness.

For the second example, take independent increments $D_1,D_2$, each uniform on $[-1,3]$, under the physical measure $P$, and a trivial filtration at all three dates. Let $S_0=3$, $S_t=S_0+\sum_{i=1}^tD_i$, so prices stay positive. An admissible strategy is a constant pair $(h_1,h_2)$. If it is nonzero, the linear form $h_1x_1+h_2x_2$ takes both signs on open subsets of $(-1,3)^2$: for example, sufficiently small positive multiples of $(h_1,h_2)$ and its negative both lie in that square. Independence and positive density give each sign positive probability. Hence the original gain model is arbitrage-free. Its projection under $P$, however, is
\[
\E_P[S_t\mid\F_t]=3+t,
\]
because $\E_PD_i=1$. The constant strategy $(1,1)$ makes a sure gain of 2 in this projected market. A projection under an arbitrary physical measure need not preserve absence of arbitrage. The main paper instead chooses $Q_0\in\M_S$.

\subsection{The continuous-time illustration and its boundary}
Kabanov and Stricker~\cite[Section~5]{KS2006} also give a continuous-time example. Here a Brownian motion $W$ must be defined for all $t\geq0$, since $e^t$ can exceed a prescribed finite trading horizon. Put $\phi(x)=\pi+\arctan x$ and
\[
Y_t=\int_0^t\phi(W_s)\,dW_s+\int_0^t W_{e^s}\,ds.
\]
The second summand is a random, anticipative finite-variation process. Kabanov and Stricker establish that $Y$ cannot be a semimartingale in a filtration satisfying the usual conditions that adapts it. Their argument recovers $W$ from the quadratic variation $\int_0^t\phi(W_s)^2ds$ and then recovers future Brownian values from the finite-variation part. This semimartingale assertion is taken from that source, not from the discrete-time theorem.

For trivial information, expected gains can be computed directly. For a bounded deterministic Borel function $a$ and a finite horizon $L$, define
\[
(a\cdot Y)_L=\int_0^L a(s)\phi(W_s)\,dW_s+\int_0^L a(s)W_{e^s}\,ds.
\]
The first term is integrable and has expectation zero by the square-integrable It\^o integral: $\phi$ and $a$ are bounded. For the second,
\[
\E\int_0^L|a(s)W_{e^s}|\,ds
\leq\|a\|_\infty\int_0^L e^{s/2}\,ds<\infty,
\]
using $\E|W_t|\leq\sqrt{\E W_t^2}=\sqrt t$. Fubini's theorem and $\E W_{e^s}=0$ give zero expectation. Thus $\E[(a\cdot Y)_L]=0$, and a nonnegative terminal gain in this specified class must vanish almost surely. The projection onto the completed trivial filtration is the zero martingale. This calculation concerns bounded deterministic integrands; it neither invokes Theorem~\ref{thm:KSgeneral} in continuous time nor establishes a completeness result there.

%% file: computation.tex
\section{A mathematical procedure for the delayed hedge}\label{app:computation}
This appendix specifies the construction of Proposition~\ref{prop:binomial} entirely in terms of finite histories and arithmetic operations. It applies to path-dependent claims, including claims that distinguish histories with the same current price.

\subsection{Inputs and terminal values}
Fix $T\geq1$, $k\geq0$, $S_0>0$, and $0<d<b<u$. Set $n=(T-k)^+$ and $q=(b-d)/(u-d)$. Write $\mathcal Z_j=\{d,u\}^j$, with $\mathcal Z_0=\{\varnothing\}$. If the undiscounted payoff $\xi$ is initially given on full paths of length $T$, first require
\[
\xi(\omega)=\xi(\omega')\quad\text{whenever}\quad
(\omega_1,\ldots,\omega_n)=(\omega'_1,\ldots,\omega'_n).
\]
This is exactly $\F_T^k$-measurability, since these prefixes are its atoms. Denote the common payoff on a prefix $z\in\mathcal Z_n$ by $\xi(z)$, and put
\[
v_n(z)=\frac{\xi(z)}{b^T},\qquad
x_j(z)=S_0\prod_{i=1}^j\frac{z_i}{b},\qquad x_0(\varnothing)=S_0.
\]
The exponent in terminal discounting is $T$, including when $n<T$. If $n=0$, the payoff is constant: take $c=\xi/b^T$, $H_t=0$, and $\beta_t=c$ at every trading date. The remainder concerns $n\geq1$.

\subsection{Backward recursion and calendar-time holdings}
For $j=n,n-1,\ldots,1$ and each $z\in\mathcal Z_{j-1}$, calculate
\begin{align}
v_{j-1}(z)&=qv_j(z,u)+(1-q)v_j(z,d),\label{eq:procedurevalue}\\
a_j(z)&=\frac{b\bigl(v_j(z,u)-v_j(z,d)\bigr)}{x_{j-1}(z)(u-d)},\label{eq:proceduredelta}\\
\gamma_j(z)&=v_{j-1}(z)-a_j(z)x_{j-1}(z).
\label{eq:procedurecash}
\end{align}
All denominators are strictly positive. The two identities
\begin{equation}\label{eq:proceduresuccessors}
\gamma_j(z)+a_j(z)x_{j-1}(z)\frac{y}{b}=v_j(z,y),
\qquad y\in\{d,u\},
\end{equation}
follow from $qu+(1-q)d=b$. They determine both holdings from the two successor values.

Take initial capital $c=v_0(\varnothing)$. Before the first observed price movement, use $H_t=0$ and $\beta_t=c$ for $t\leq k$. At a later trading date $t=k+j\leq T$, the available history is $z=(Y_1,\ldots,Y_{j-1})$; take
\[
H_{k+j}=a_j(z),\qquad \beta_{k+j}=\gamma_j(z).
\]
Both holdings are $\F_{k+j-1}^k$-measurable. At an effective node $z$ of length $j-1$, their discounted cost is $v_{j-1}(z)$ by~\eqref{eq:procedurecash}. After either successor, their discounted value is $v_j(z,y)$ by~\eqref{eq:proceduresuccessors}, exactly the cost of the next holdings. This proves self-financing at every rebalance, including the first purchase after the initial cash-only dates. Induction over $j$ therefore gives
\[
W_t=v_{(t-k)^+}(Y_1,\ldots,Y_{(t-k)^+}),\qquad
V_t=b^tW_t,\qquad V_T=\xi.
\]
For nonnegative payoffs, the recursion also gives $W_t\geq0$ at every date. Shares and cash holdings themselves may be negative, consistently with the unconstrained trading model.

\subsection{A check with nonzero interest}
Take $T=2$, $k=1$, $S_0=100$, $u=6/5$, $d=4/5$, $b=11/10$, and $\xi=\ind_{\{Y_1=u\}}$. Then $n=1$, $q=3/4$, and
\[
v_1(u)=\frac{100}{121},\quad v_1(d)=0,\quad
c=\frac{75}{121},\quad
H_1=0,\quad\beta_1=\frac{75}{121},\quad
H_2=\frac1{44},\quad\beta_2=-\frac{200}{121}.
\]
The rebalance at date one costs $\beta_2+100H_2=75/121$ in discounted units, equal to the available wealth. At date two,
\[
X_2=\begin{cases}1200/11,&Y_1=u,\\800/11,&Y_1=d,\end{cases}
\qquad
\beta_2+H_2X_2=\begin{cases}100/121,&Y_1=u,\\0,&Y_1=d.\end{cases}
\]
Multiplication by $B_2=121/100$ gives exactly $\xi$ on all four full paths. Neither the positions nor the payoff use the unobserved second return. This calculation also checks that the auxiliary risky price at maturity is $B_2X_2$, not $S_1$ when $b\ne1$.